\documentclass[letterpaper, 10 pt, conference]{ieeeconf}  
\IEEEoverridecommandlockouts                              

  \let\labelindent\relax
\usepackage{amsmath} 
\usepackage{amssymb}  
\usepackage{mathrsfs} 
\usepackage{tabulary}
\usepackage[english]{babel}
\usepackage{blindtext}
\usepackage[font=footnotesize]{caption}
\usepackage[makeroom]{cancel}
\usepackage{amsfonts}
\usepackage{subfig}
\usepackage{amsthm}  
\usepackage{amssymb,hyperref}
\usepackage{graphicx}
\usepackage{dsfont}
\usepackage{algorithm2e}
\RestyleAlgo{ruled}
\usepackage{algpseudocode}
\usepackage{tabularx}
\usepackage{float}
\usepackage{cite}
\usepackage{enumitem} 
\usepackage{bbm} 
\usepackage[table]{xcolor}
\usepackage[strict]{changepage}

\newcommand{\vect}[1]{\boldsymbol{#1}}
\newcommand{\mat}[1]{\boldsymbol{#1}}

\DeclareMathOperator*{\argmax}{\rm{argmax}} 
\renewcommand{\eqref}[1]{Eq.~(\ref{#1})}  
\newtheorem{remark}{Remark}
\newtheorem{theorem}{Theorem}
\newtheorem{lemma}{Lemma}
\newtheorem{assumption}{Assumption}

\newtheorem{definition}{Definition}

\newtheorem{proposition}{Proposition}

\newtheorem{problem}{Problem}

\allowdisplaybreaks

\title{\LARGE \bf
On a Closed-Loop  Controller for the Coevolutionary Model\\of Actions and Opinions via Broadcasting Information
}

\author{Roberta Raineri, Mengbin Ye, and Lorenzo Zino
\thanks{R. Raineri and L. Zino are with the Department of Electronics and Telecommunications, Politecnico di Torino, Turin, Italy (\texttt{\{roberta.raineri,lorenzo.zino\}@polito.it}). R. Raineri is with Modelway, Turin, Italy. 
M. Ye is with the Adelaide Data Science Centre, Adelaide University, Adelaide, Australia (\texttt{ben.ye@adelaide.edu.au}). M.~Ye is supported by the Australian Government via the Australian Research Council (DE250100199). The work of L. Zino is supported by Fondazione CRT and Politecnico di Torino via the “Bando Internazionalizzazione della Ricerca 2025” (CUP: E13C25002890005).
}%
}

\begin{document}

\maketitle
\thispagestyle{empty}

\begin{abstract}
We deal with controlling a complex social network in which agents have actions and opinions that coevolve, mutually influencing one another. We consider an input consisting in broadcasting information to a target set of agents with the objective of steering the population, initially at a consensus, to a different consensus. For a constant input, we derive a monotone convergence result, building on which we design an algorithm that determines whether a target set is sufficient to achieve the  objective and an effective heuristic to optimize the target set. Then, we introduce a feedback control law that, using information on the state of the system, dynamically revises the target set,  reducing the  effort needed to achieve the objective while guaranteeing convergence to the desired consensus state. 
\end{abstract}

\section{Introduction}\label{sec:intro}

Many real-world social phenomena involve a population of individuals who make decisions on (often binary) actions to adopt (e.g., use public transportation or private cars) on the basis of several factors, including the individual's opinion (i.e., the support or preference for a given action) and social pressure, leading to potential misalignment between an individual's action and their opinion~\cite{centola2005emperor}. Classical mathematical models of social dynamics focus separately on decision-making and opinion formation~\cite{friedkin2015_socialsurvey,proskurnikov2017tutorial,montanari2010spread_innovation}), overlooking the complexity that arises due to their coupling~\cite{gavrilets2017collective}. 

Continuous-opinion discrete-action (CODA) models constituted a seminal step to address this limitation. In the CODA framework~\cite{Martins2008coda,Ceragioli2018quantized}, the process is driven by the opinion dynamics, whilst actions are a direct quantization of an individual's opinion. This limits the possibility to capture a misalignment between actions and opinions. To overcome this limitation, in~\cite{zino2020chaos} a coevolutionary model of actions and opinions was proposed in which individuals simultaneously update their actions and opinions, each a separate independent state. The model utilizes a game-theoretic formulation that captures social influence and peer pressure, and also allows for the analytical study of the model; see~\cite{Zino2020cdc,Hassan2023tac}.

Building on the coevolutionary model of actions and opinions~\cite{Hassan2023tac}, one can study how to guide the collective emergent behavior using approaches that exploit the complex coupling of opinions and actions, including the possibility (as in this paper) of influencing just the opinions to steer actions and opinions. A first effort focused on the problem of steering the whole population from one consensus state to another by directly controlling the state (action and/or opinion) of a subset of nodes~\cite{raineri2025_tcns}. For this scenario, an algorithm to select the minimal subset of nodes to control was proposed in~\cite{raineri2025_tcns} and shown to be effective, but it has two main limitations. First, it requires direct control of some agents, which is often infeasible or unethical. Second, it assumes that the selected agents remain permanently controlled, leading to successful interventions but potentially unnecessary long-term costs. 

This work is specifically aimed at overcoming these two limitations. We incorporate an external control input into the coevolutionary model of actions and opinion from~\cite{Hassan2023tac}, representing information that is broadcast to select individuals of the population, e.g., by a media source~\cite{Li2020} or a recommender agent~\cite{Rossi2022,Sprenger2024}. In other words, we can directly control the content that the input shares and to whom. Hence, our  goal is to understand which nodes to target with this information. To this aim, we develop an algorithm to determine whether a proposed target set is sufficient to guarantee the desired outcome, and an effective heuristic to determine the minimal number of nodes that needs to be targeted, since the problem is computationally intractable due to its combinatorial nature. Then, we further reduce the control effort by addressing the second limitation described above, i.e., designing a dynamic control policy that reduces the number of targeted nodes over time. To this aim, we design our control policy in a closed-loop fashion, using the algorithm described earlier to initialize the control input, and then leveraging state information to adaptively reduce it.

This work has three main contributions. First, we encapsulate in the coevolutionary model~\cite{Hassan2023tac} an external control input. Second, for a constant control input, we prove convergence, we propose an algorithm to determine whether a constant control input is able to achieve the desired goal, and a  heuristic to determine the optimal input. Third, we design a closed-loop control scheme that guarantees convergence and dynamically revises the target set to reduce the control effort.

\section{Model and Problem Statement}\label{sec:model}

\subsection{Coevolutionary Model with an External Input}
\label{sec:dynamics}

\textit{Players.} We consider a population  $\mathcal V = \{1,2,\hdots,n\}$ of $n$ individuals (players). Each $i\in\mathcal V$ is associated with a two-dimensional state variable $\vect{z_i}(t) = (x_i(t),y_i(t))\in\{-1,+1\}\times[-1,+1]$ that is updated at discrete time $t$. The first entry, $x_i(t)\in\{-1,+1\}$, represents the (binary) \emph{action} of  $i$ at time $t$; the second entry, $y_i(t)\in[-1,+1]$, their \emph{opinion} on the action. In other words, $y_i(t)=-1$ means that $i$ is totally in favor of action $-1$;  $y_i(t)=+1$ that $i$ fully supports action $+1$. Actions and opinions of the individuals are gathered in vectors $\vect{x}(t)\in\{-1,1\}^n$ and  $\vect{y}(t)\in[-1,1]^n$, and the state of the system is compactly denoted by the joint $2n$-dimensional state vector $\vect{z}(t):=(\vect{x}(t),\vect{y}(t))\in\{-1,1\}^n\times [-1,1]^n$. Consistent with the game-theoretic formulation in~\cite{Hassan2023tac}, given an individual $i\in\mathcal V$, we define $\vect{z_{-i}}:=(\vect{x_{-i}},\vect{y_{-i}})\in\{-1,1\}^{n-1}\times [-1,1]^{n-1}$ as the $(2n-2)$-dimensional vector with the state of all others. 

Individuals interact on a two-layer, strongly connected, weighted social network $\mathcal G=(\mathcal V,\mathcal E_A, \mat A,\mathcal E_W, \mat W)$. First, they observe others' actions on an influence layer, where $(i,j)\in\mathcal E_A$ means that $i$ can observe $j$'s action and the nonnegative weight $a_{ij}$ captures the corresponding influence ($a_{ij}>0\iff (i,j)\in\mathcal E_A$). Second, they share their opinions on a communication layer, where $(i,j)\in\mathcal E_W$ means that $i$ receives information from $j$'s opinion, captured by the nonnegative weight $w_{ij}$ ($w_{ij}>0\iff (i,j)\in\mathcal E_W$). The two weight matrices are stochastic, i.e.,   $\mat A\vect1=\mat W\vect{1}=\vect{1}$, where $\vect1$ is the $n$-dimensional all-$1$ vector.

We introduce an external input, modeled by an external agent that can represent different  entities (e.g., a media source or a social bot). This external agent  broadcasts information $\iota(t)\in[-1,1]$ to a subset of nodes. Specifically, we introduce a vector $\vect{\gamma}(t) \in [0,1]^n$, where the entry $\gamma_i(t)$ is the effort that the controller places in targeting agent $i$ at time $t$ with the information $\iota(t)$ from the external agent. We denote by 
    $\mathcal C(t):=\left\{i:\gamma_i(t)>0\right\}$ 
the \emph{target set} of agents receiving the broadcast at time $t$. Given the limited human ability to process information, we assume that if agent $i$ receives $\gamma_i(t)$ information from the external input, then the information received from peers is reduced by the same factor. In other words, the external input replaces part of the information shared on the communication layer with the information broadcast by the external agent. 


\textit{Utility function.} 
A function $u_i(\vect{z_i},\vect{z_{-i}})$ was proposed in \cite{Hassan2023tac} to represent the utility that individual $i$ receives for selecting an action and opinion pair $\vect{z_i}=(x_i,y_i)$ when the state of the others is $\vect{z_{-i}}$. We re-write it 
to explicitly include  the information shared by the external input:
\begin{align}
    &u_i(\vect{z_i},\vect{z_{-i}})=\tfrac{\lambda_i (1-\beta_i)}{2}\sum\nolimits_{j \in \mathcal V} a_{ij} \big[ (1-x_j)(1-x_i) \nonumber\\&+ (1+x_j)(1+x_i) \big]\hspace{-.1cm} -\hspace{-.1cm}\beta_i (1-\lambda_i)(1\hspace{-.05cm}-\hspace{-.05cm}\gamma_i)\displaystyle \sum_{j \in \mathcal V}{w}_{ij}(y_i-y_j)^2\nonumber\\&-\beta_i (1-\lambda_i)\gamma_i(y_i-\iota)^2-\lambda_i\beta_i(x_i-y_i)^2,\label{eq:utility}
\end{align}
where $\lambda_i ,\beta_i\in(0,1]$ capture the weights given to actions observed and opinions exchanged, respectively. For the sake of readability, we have omitted to explicitly show the time-dependency of $\gamma_i(t)$ and $\iota(t)$. In \eqref{eq:utility}, the first term accounts for the individuals' tendency to coordinate actions, typical of coordination games~\cite{montanari2010spread_innovation,ye2021nat}. The second and third terms capture the influence of others' opinions and of the external source, respectively,  typical of opinion dynamics~\cite{friedkin1990_FJsocialmodel,proskurnikov2017tutorial}. The fourth term describes an individual's desire for consistency between their action and opinion. 

\textit{Update Process.} We define the set $\mathcal R(t)\subseteq \mathcal V$ as those individuals that simultaneously become active at time $t$. Each individual $i\in\mathcal R(t)$ updates their action and opinion, aiming to maximize their utility via a joint best-response with respect to the function in \eqref{eq:utility}. This yields:
\begin{equation}\label{eq:update}x_i(t+1),y_i(t+1)\hspace{-.1cm}=\hspace{-.1cm}\left\{\hspace{-.1cm}\begin{array}{ll}\displaystyle\argmax{u_i(\vect{z_i}(t),\vect{z_{-i}}(t))} &i\in\mathcal R(t),\\
x_i(t),y_i(t)&i\notin \mathcal R(t),\end{array}\right.
\end{equation}
with the standard tie-breaker convention of $\vect{z_i}(t+1)=\vect{z_i}(t)$ when $\argmax u_i(\vect{z_i}(t),\vect{z_{-i}}(t))$ comprises multiple elements, adopted from~\cite{Hassan2023tac,raineri2025_tcns}. We  make the following standard assumption on the activation sequence. 

\begin{assumption}\label{a:activation}
   There exists a constant $T<\infty$ such that $\cup_{s=0}^{T-1} \mathcal R(t+s)=\mathcal V$, for any $t\geq 0$.
\end{assumption}

From the utility function in~\eqref{eq:utility}, an explicit closed-form expression for~\eqref{eq:update} can be derived, similar to~\cite[Proposition~1]{Hassan2023tac}, reason for which we omit the explicit proof.

\begin{proposition}\label{prop:dynamics}
Individual $i\in\mathcal R(t)$ updates their state as:
\begin{subequations}\label{eq:dinamics}
\begin{align}
\label{x-dinamic}
x_i(t+1) =&s(\vect{z}(t)),\\
\label{y-dinamic}
y_i(t+1)=&(1-\lambda_i)(1-\gamma_i(t))\sum\nolimits_{j \in \mathcal{V}} w_{ij}  y_j(t)\nonumber\\& +(1-\lambda_i)\gamma_i(t)\iota(t)+\lambda_i s(\vect z(t)) 
\end{align}
\end{subequations}
with 
$$s(\vect z(t))=\begin{cases}
    +1 \qquad &\text{if}\; \delta_i(\vect z(t)) >0, \\
    -1 \qquad &\text{if}\; \delta_i(\vect z(t)) <0, \\
    x_i(t) \qquad &\text{if}\; \delta_i(\vect z(t)) =0,
\end{cases}
$$
and 
$\delta_i(\vect z(t)) = 2 \beta_i (1-\lambda_i)[(1-\gamma_i(t))\sum_j w_{ij}y_j(t)+\gamma_i(t)\iota(t)]+(1-\beta_i)\sum_j a_{ij}x_j(t)$.
\end{proposition}

When $\vect\gamma(t)=\vect 0$ for all $t\geq 0$, the update rule in Proposition~\ref{prop:dynamics} reduces to the one of the uncontrolled coevolutionary dynamics with no prejudice, extensively studied in~\cite{raineri2024}. 
Before formulating our control problem, it is worth noticing that the two consensus states $\vect{x}=\vect{y}=-\vect{1}$ and $\vect{x}=\vect{y}=\vect{1}$ are always equilibria in the absence of any control action.

\begin{proposition}[Proposition 3 from~\cite{raineri2024}]\label{prop:equilibria}
If $\vect\gamma(t)=\vect 0$ for all $t\geq 0$, then  \eqref{eq:update}  has at least two equilibria: $\vect{x}=\vect{y}=-\vect{1}$ and $\vect{x}=\vect{y}=\vect{1}$. They are the only two equilibria in which the action vector is at a consensus ($x_i=x_j$, $\forall\,i,j\in\mathcal V$). 
\end{proposition}

\subsection{Problem Statement}

We consider the problem of steering a population from one consensus state to the other. 
Without any loss in generality, we assume that the initial consensus is $\vect{x}(0)=\vect{y}(0)=-\vect{1}$, and thus the goal is to reach $\vect{x}=\vect{y}=\vect{1}$. An authority can execute a controller that decides which nodes to target with the external input, and what information to broadcast to them. Concerning the latter, since our goal is to maximize $\vect{z}(t+1)$, it is always convenient to set $\iota(t)=1$.



Hence, the control problem ultimately reduces to designing a control law $\Gamma:t\to\vect{\gamma}(t)$ that determines the set of target nodes $\mathcal C(t)$ in order to guarantee convergence to the all-$1$ consensus. 
In this work, for the sake of clarity, we consider the further simplifying assumption that the controller can only choose whether to target or not a node (and when), but not how much. In other words, we assume binary entries $\gamma_i(t)\in\{0,1\}$, to denote which nodes are targeted at time $t$. Its weakening is set for future extensions.
In the following, we  summarize all the assumptions made for this study.

\begin{assumption}
\label{a:model}
Consider a two-layer network $\mathcal G=(\mathcal V,\mathcal E_A,\mat{A},\mathcal E_W,\mat{W})$ with $\mat{A}$ and $\mat{W}$ stochastic and irreducible. Let $\vect{x}(0)=\vect{y}(0)=-\vect{1}$, and assume that $\vect{\gamma}(t)\in\{0,1\}^n$ and $\iota(t)=1$ for all $t\geq 0$.
\end{assumption}

\begin{remark}\label{rem:set}
    Under Assumption~\ref{a:model}, there is a one-to-one correspondence between $\vect{\gamma}(t)$ and the target set $\mathcal C(t)$. Hence, we can freely see a control law in terms of designing the control input $\vect\gamma(t)$ or  the target set $\mathcal C(t):=\{i:\gamma_i(t)=1\}$.
\end{remark}


Due to the generality of Assumption~\ref{a:activation}, a given control law $\Gamma$ and $\vect z(0)$ can produce different trajectories $\vect z(t)$ depending on the  activation sequence. For this reason, we introduce the probability that the dynamics in \eqref{eq:dinamics} with control law $\Gamma$ converges in finite time to $\vect{x^*}=\vect{1}$, computed over the probability space generated by the activation sequence as 
    $\phi(\Gamma)=\mathbb P[\exists\,T<\infty:\vect{x}(t)=\vect{1}, \forall\,t\geq T]$. 
Hence, our  objective is to design $\Gamma$  to ensure $\phi(\Gamma)=1$. Finally, to compare different control laws, it is convenient to define a cost associated with the control effort as follows.

\begin{definition}\label{def:cost}
  The total control effort up to time $T$ is
    \begin{equation}\label{eq:cost}
        J=\int_0^T\vect{\gamma}(t)^\top\vect{1}dt.
    \end{equation}
\end{definition}

\begin{problem}\label{problem2}
       Design a control law $\Gamma:t\to\vect{\gamma}(t)$ such that, under Assumptions~\ref{a:activation}--\ref{a:model}, $\phi(\Gamma)=1$, with the aim of reducing the total control effort in \eqref{eq:cost}.
\end{problem}

The meaning of reducing the total control effort in Problem~\ref{problem2} is context-dependent, and precise details will be given in the two sections below, that are devoted to designing cost-effective constant and feedback control laws, respectively.


\section{Constant Control Input Approach}\label{sec:constant}

\begin{assumption}\label{a:constant}
  Assume $\vect{\gamma}(t)=\vect{\gamma}$, for all $t \geq 0$.
\end{assumption}

Since each constant control law $\Gamma$ is associated with a (constant) target set $\mathcal C_{\Gamma} := \{i:\gamma_i=1\}$, we can simply write $\phi(\Gamma)$ in terms of the target set, denoting it as $\phi(\mathcal C_{\Gamma}):=\phi(\Gamma)$. Furthermore, observe that the total control effort is proportional to the number of targeted nodes, i.e, $J=|\mathcal C_\Gamma|T$. Hence, Problem~\ref{problem2} reduces to finding the minimal subset of nodes $\mathcal C$ to target such that $\phi(\mathcal C)=1$. First, we state a result that guarantees convergence of the dynamics under Assumptions~\ref{a:activation}--\ref{a:constant}. The proof is reported in the Appendix.

\begin{proposition}\label{prop:convergence}
   The dynamics in \eqref{eq:update} under Assumptions~\ref{a:activation}--\ref{a:constant} is monotonically non-decreasing and converges to an equilibrium point $(\vect x^*, \vect y^*)$. Specifically, $\vect{x}(t)$ converges to $\vect{x^*}$ in finite time, and $\vect{y}(t)$ converges to $\vect{y^*}$ asymptotically.
\end{proposition}

Building on Proposition~\ref{prop:convergence} and inspired by~\cite{raineri2025_tcns}, we address Problem~\ref{problem2} in two steps. First, we propose Algorithm~\ref{alg} to check if a specified targeting set $\mathcal C$ is sufficient to guarantee that $\phi(\mathcal C)=1$. This is because Proposition~\ref{prop:convergence} only establishes convergence to an equilibrium $(\vect x^*, \vect y^*)$, which for some choices of $\mathcal C$ is not guaranteed to be the all-$1$ consensus equilibrium. Second, we define Algorithm~\ref{alg:optimal_C} as an effective heuristic to find the minimum number of nodes to target.

Algorithm~\ref{alg} starts from a candidate set of nodes with action equal to $1$ at the equilibrium, $\mathcal A(0)$. At each step $k$, we construct the unique candidate equilibrium that has action equal to $+1$ only for those nodes in $\mathcal A(k)$, namely \eqref{eq:candidate}, and we check if the candidate equilibrium is actually an equilibrium, by computing $\delta_i$. If a node $i\notin \mathcal A(k)$ has $\delta_i>0$, then the candidate is not an equilibrium since $i$ will eventually switch action to $+1$. In this case, the node is added to the set $\mathcal A(k+1)$ and the iteration index $k$ increased by one. This procedure ultimately converges (in no more than $n$ steps) to the set of all  nodes that will switch to $+1$. Hence, $\phi(\mathcal C)=1$ if and only if $\mathcal A_f=\mathcal V$. 
The theoretical arguments that guarantee the well-functioning of the algorithm are based on Proposition~\ref{prop:convergence}, which guarantees that nodes that switch to $+1$ action never switch back. The  details are omitted, since they closely follow~\cite[Theorem~3]{raineri2025_tcns}. 

\begin{algorithm}[t]
\caption{Equilibrium computation
\label{alg}
}
\KwData{$\mat{W}, \mathcal C, {\vect x_0}, \vect \lambda$ and $\vect \beta$}
\KwResult{$\mathcal A_f:=\mathcal A(k)$}
$k \gets 1;\; \mathcal A(0) \gets \emptyset;\; 
\mathcal A(1) \gets \{i \in \mathcal{V}:{x_0}_i = +1\}$\; $\gamma_i \gets +1$ $\forall i \in \mathcal{C}$;\,
$\mat{M}\gets(\mat{I}-(\mat I- \text{diag}(\vect\lambda))\mat W)^{-1}$\;
\While{$\mathcal A(k) \neq \mathcal A(k-1)$}{
\begin{subequations}\label{eq:candidate}\begin{align}\label{x-candidate}{\hat x}_i\gets \begin{cases}
1 \; \text{for all} \; i \in \mathcal{A}(k)\\
-1 \; \text{for all} \; i \notin \mathcal{A}(k);
\end{cases}\\
\label{y-candidate}{\hat y}_i\gets \begin{cases}
(\mat M \text{diag}(\vect \lambda) \vect{\hat x})_i \; \text{for all} \; i \notin \mathcal{C}\\
1-\lambda_i+\lambda_i\hat{x}_i  \; \text{for all} \; i \in \mathcal{C};
\end{cases}\end{align}
\end{subequations}
$k \gets k+1;\;
\mathcal A(k) \gets \mathcal A(k-1)$\;
check    \For{$i \in \mathcal V \And i \notin\mathcal A(k)$}
    {
        \If{$\delta_i(\vect{\hat x},\vect{\hat y})>0$} 
        {$\mathcal A(k) \gets \mathcal A(k) \cup \{i\}$;}
    }
}
\end{algorithm}

\begin{algorithm}
\caption{Optimal targeting set identification \label{alg:optimal_C}
}
\KwData{$\mat{W}, \vect \lambda$, $\vect \beta$, $\varepsilon$,  $n$, {$K$}}
\KwResult{$\hat\mu_X$, empirical invariant distribution of $X(k)$ }
$k \gets 1;$ 
${X}(k) \gets \mathcal{V}$;

{\While{$k <$ {$K$}}{
$k \gets k+1$; 
$X(k) \gets X(k-1)$\;
Choose at random a node $r \in \mathcal V$\;
\If{$r \in X(k)$ {\bf and} Algorithm \ref{alg} with $\mathcal C=X(k-1)\setminus\{r\}$ yields $\mathcal{A}_f = \mathcal{V}$}
{$X(k) \gets X(k) \backslash \{r\}$ \;
}
\Else{
$X(k)\gets  X(k) \cup \{r\}$ with probability $\varepsilon$ 
}
 }
 }
\end{algorithm}
Algorithm~\ref{alg:optimal_C} starts by checking that the problem is feasible, i.e., $\phi(\mathcal V)=1$. Then, we define a Markov chain $X(k)$ over the space of admissible control sets $\bar{\mathscr{C}}$, i.e., ${\mathcal{C}} \in \{0,1\}^n$ such that $\phi(\mathcal{{C}})=1$, initialized as  $X(0)=\mathcal V$. At the \mbox{$k$-th} iteration, given the current state $X(k)$, we select a node $r\in\mathcal V$ uniformly at random. If the node belongs to the current state and  Algorithm~\ref{alg} guarantees that $\phi(X(k)\setminus\{r\})=1$, the Markov chain jumps to $X(k)\setminus\{r\}$. If the node does not belong to the current state, the chain jumps to $X(k)\cup\{r\}$ with probability $\varepsilon$, and is not update otherwise. This procedure allows the algorithm to explore the space of control sets with $\phi=1$, obtaining the following result. 

\begin{theorem}\label{th:1}
The invariant distribution $\mu_\varepsilon \in [0,1]^{\mathscr{\bar{C}}}$ of $X(t)$ is such that $\lim_{\varepsilon \searrow 0} \mu_\varepsilon = \mu$ where $\mu$ is the uniform distribution on
    the solution of Problem~\ref{problem2}, under Assumption~\ref{a:activation}--\ref{a:constant}.
\end{theorem}
\begin{proof}
Proposition~\ref{prop:convergence} and the supermodularity of the game  (which can be shown analogously to~\cite[Proposition~1]{raineri2025_tcns}), ensures that the assumptions of~\cite[Theorem~2]{Como2022supermodular} are satisfied. Application of~\cite[Theorem~2]{Como2022supermodular} yields that $\mu_\varepsilon$ converges, as $\varepsilon \searrow 0$, to the uniform probability measure $\mu$ over the admissible control sets $\mathcal{C}$ that minimize the total control effort. For a detailed derivation applying similar arguments, though for a different algorithmic setting, see~\cite[Theorem~4]{raineri2025_tcns}.
\end{proof}


\section{Closed-loop Input Approach}\label{sec:feedback}
A major limitation of the algorithm proposed in Section~\ref{sec:constant} is in its static nature. In fact, as the network approaches $\vect x = \vect y = \vect 1$, maintaining the same set of targeted nodes may become unnecessary and may increase the overall control effort.
To design a time-varying controller, however, we should keep in mind that the applicability of 
Algorithm~\ref{alg} is centered around selecting target nodes that ensure the monotonicity of $\vect{x}(t)$ over time. If the target nodes are dynamically adjusted (e.g., an initially targeted node is no longer targeted after some time $\bar t$), we have in general no guarantees that $\vect x(t)$ remains monotone after $\bar t$, calling for the careful design of a time-varying scheme. 


Our revision algorithm is based on the following idea. Starting from an initial target set, we ask whether a particular node $i$ in the target set can be removed at some $\bar t> 0$, without compromising the fact that $\phi=1$. There are two sequentially evaluated sufficient conditions that enable this determination. First, we restrict the admissible revisions (nodes to remove from the target set) to only those ones that are guaranteed to not result in a future flip back of any node actions from $+1$ to $-1$. That is, removing nodes from the target set at time $\bar t$ does not change the monotonically non-decreasing nature of $\vect x(t)$ for $t\geq \bar t$. Second, among the admissible sets, we only allow revisions for which Algorithm~\ref{alg}, when re-initialized with the current action profile at $\bar t$, provides a guaranteed convergence to the desired equilibrium of $\vect x^* = \vect y^* = \vect 1$.

To guarantee the first condition, given $\bar t\geq 0$, we define $\mathcal{A}(\bar t):=\{i \in \mathcal{V}:x_i(\bar t)=+1\}$. Moreover, given a node $i\in\mathcal C(\bar t)$, we define the vector $ \vect {\tilde y}(\bar t)$, with components
\begin{equation}\label{eq:tilde_y}{\tilde y_i}(\bar t):=\min_{\ell \in \mathcal{V}} y_\ell(\bar t)(1-\lambda_i)(1-\tilde\gamma_i)+(1-\lambda_i)\tilde\gamma_i + \lambda_i x_i(\bar t).\end{equation} 

\begin{proposition} \label{prop:delta}
If $\delta_{i}(\vect x(\bar t), \vect {\tilde y}(\bar t))\geq 0$ for all $i\in\mathcal A(\bar t)$, with $\vect {\tilde y}(\bar t)$ from \eqref{eq:tilde_y}, then it holds that $\mathcal{A}(t)\supseteq \mathcal A(\bar t)$, for all $t>\bar t$ under the control policy $\mathcal C(t)=\mathcal C(\bar t)\setminus\{i\}$, for all $t>\bar t$. 
\end{proposition}
\begin{proof}
     We proceed by contradiction. 
     Assume that the claim is false. Hence, it necessarily exists a time $t^*\geq0$ and  $i\in\mathcal A(\bar t)$ such that, it holds that $\delta_i(\vect{x}(\bar t+t^*), \vect y(\bar t+t^*))<0$, and $\delta_j(\vect{x}(\bar t+s), \vect y(\bar t+s))\geq 0$ for all $s\in\{0,\dots,t^*-1\}$ and  $j\in\mathcal A(\bar t)$. In other words, $i$ is the first node that can switch back from $+1$ to $-1$, if activated at time $\bar t+t^*$. This implies that $x_j(\bar t+t^*)=+1$ for all $j\in\mathcal A(\bar t)$, and so $x_j(\bar t+t^*)\geq x_j(\bar t)$ for all $j\mathcal\in \mathcal V$.

    The key observation is now that by construction $\tilde{y}_j(\bar t)\leq y_j(\bar t+t^*)$ for any  $j\in\mathcal V$. In fact, \eqref{y-dinamic} posits that the opinion of those individuals who update is a convex combination of the weighted average of their neighbors (which is bounded by the minimum opinion $\min_{i\in\mathcal V} y_i(\bar t)$), the input, and the action. Using this recursively from $\bar t$ to $\bar t+t^*$ and using the fact that the input is assumed to be constant for $s\geq \bar t$ and the action vector at time $\bar t+t^*$ is not smaller than the one at time  $\bar t$ (since the first switch from $+1$ to $-1$ occurs not before time $\bar t+t^*+1$), we obtain \eqref{eq:tilde_y}.
    
    Thus, on the one hand, $\delta_{i}(\vect x(\bar t+t^*), \vect {y}(\bar t+t))<0$, while,  on the other hand, $\delta_{i}(\vect x(\bar t+t), \vect {y}(\bar t+t)) \geq \delta_{i}(\vect x(\bar t), \vect {\tilde y}(\bar t))>0$, 
    which are in contradiction, yielding the claim.
\end{proof}

In other words, a targeted node~$i$ that satisfies the hypotheses of Proposition~\ref{prop:delta} can be removed from the target set at time $\bar t$ and it is still guaranteed that the dynamics $\vect{x}(t)$ remains monotonically non-decreasing for $t\geq \bar t$. This is key, because it now allows us to use Algorithm~\ref{alg} to check whether the controlled dynamics (now with constant target set equal to $\mathcal C(\bar t)\setminus\{i\}$) still converges to the desired equilibrium when the initial state is $\vect{x_0}=\vect{x}(\bar t)$ --- this addresses the second sufficient condition discussed above Proposition~\ref{prop:delta}. 

Building on this intuition, we design a closed-loop control law $\Gamma$ for Problem~\ref{problem2}, summarized in Algorithm~\ref{alg:final}. 
We start by defining a sequence of revision times $\mathcal T=\{T_1,T_2,\dots\}$. In general this sequence can be  deterministic, stochastic, or even state-dependent. We assume that the sequence is discrete and countable, so that for any $t\geq 0$, there exists $T_s\in\mathcal T$ such that $T_s\geq t$. Then, we initialize the control input as $\vect{\gamma}(t)=\vect{\gamma_0}$ for all $t\in[1,T_1]$, where $\vect{\gamma_0}$ is the constant control input solution of Algorithm~\ref{alg:optimal_C} with $\vect x_0 = \vect x(0)$.

At each subsequent revision time $T_k$, we use Proposition~\ref{prop:delta} to detect a node that can be potentially removed from the target set, without losing the monotonicity of $\vect{x}(t)$, and Algorithm~\ref{alg} to determine whether we can stop targeting that node without losing convergence. 
If we can stop targeting  $i$,  we revise the control input setting $\gamma_i(t)=0$ for all $t\geq T_k+1$,  keeping the other entries unchanged. 

\begin{algorithm}
\caption{Feedback control law
\label{alg:final}
}
\KwData{$\vect{x}(t),\vect{y}(t),\vect{\gamma}(t),t,\mathcal T, \mat{W},  \vect \lambda$ and $\vect \beta$}
\KwResult{$\vect \gamma(t+1)=\vect{\tilde{ \gamma}}$, closed-loop control input}
\If{$t\in\mathcal T$}{
 $\mathcal{A}\gets \{i \in \mathcal{V}: x_i=+1\}$; 
 $\mathcal{C}\gets \{i \in \mathcal{V}: \gamma_i=+1\}$\;
\For{$i \in \mathcal{C}$}{
$\vect{\tilde \gamma}\gets \vect{\gamma}(t)$;\,
$\tilde{\gamma}_i \gets 0$\;
$\tilde y_j \gets \min_k(y_k)(1-\lambda_j)(1-\tilde{\gamma}_j)+(1-\lambda_j)\tilde{\gamma}_j+\lambda_jx_j \; \forall j \in \mathcal{V}$\;
$\delta_j \gets 2\beta_j(1-\lambda_j)[(1-\tilde{\gamma}_j)\sum_k W_{jk}\tilde y_k + \tilde{\gamma}_j]+(1-\beta_j)\sum_k W_{jk}x_k$\;

        \If{$\delta_j>0 \; \forall j \in \mathcal{A}$} 
        {\If{Algorithm \ref{alg} with $\mathcal{C}=\mathcal{C}\backslash\{i\}$ and initial state $\vect x$ gives $\mathcal A_f=\mathcal V$}{$\gamma_i(t)\gets 0$\;
        {\bf break}}}
    }
}
\end{algorithm}

This algorithm ultimately yields a closed-loop control input that is piece-wise constant and, at each revision time $T_k$, the resulting $\vect\gamma(t)$ is monotonically non-increasing, guaranteeing that the total control effort is not larger than one with a constant control input, as summarized in Theorem~\ref{th:2}.


\begin{theorem}\label{th:2}
    The feedback control in Algorithm~\ref{alg:final} initialized using Algorithm~\ref{alg:optimal_C} solves Problem~\ref{problem2} under Assumptions~\ref{a:activation}--\ref{a:model}, and the total control effort in \eqref{eq:cost} is not larger than the one under constant control input computed with Algorithm~\ref{alg:optimal_C}. Moreover, there exists $\bar T$ such that $\vect{\gamma}(t)=\vect{0}$ for all $t\geq \bar T$.
\end{theorem}
\begin{proof}
     The proof is a consequence of the monotonicity of $\vect{x}(t)$ under the closed-loop control law in Algorithm~\ref{alg:final} implied by Proposition~\ref{prop:delta} and by the fact that Algorithm~\ref{alg} guarantees that $\phi(\vect\gamma)=1$. The fact that the control input is monotonically non-increasing, implies that \eqref{eq:cost} is not increased. Finally, we prove that there exists $\hat t$ s.t. $\vect\gamma(t)=\vect{0}$ for all $t\geq \hat t$ by contradiction. Assume the statement is false. Hence there exists $T_k$ such that $\vect\gamma(t)=\vect{\bar\gamma}\neq 0$ for any $t\geq T_k$. By design, there exists a finite time $\bar t$ such that $\vect{x}(t)=\vect 1$  for any $t\geq\bar t$ and $\vect{y}\to\vect{1}$ asymptotically. This implies that there exists $\hat t$ such that $\min_{i\in\mathcal V}y_i(t)> 0$ for all $t\geq \hat t$. This implies that $\tilde y_j(t)> 0$, for all $t\geq \hat t$ and consequently $\delta_j(\vect{x}(t),\vect{\tilde y}(t))\geq \delta_j(\vect{1},\vect{0})>0$ for all $t\geq \hat t$. Hence, at first revision time $T_s>\hat t$, any node $i$ with $\gamma_i=+1$ satisfies the conditions in Proposition~\ref{prop:delta}, and thus, due to Algorithm~\ref{alg:final}, $\vect\gamma(t)\neq\vect{\bar\gamma}$ for $t> T_s$, which contradicts our assumption.
\end{proof}

\section{Simulations}\label{sec:simulations}


We perform simulations in a setting in which the update dynamics is stochastic. Given a simulation time-horizon $T_{\rm s}$, we generate an activation matrix $\mat R\in\{0,1\}^{n\times T_{\rm s}}$, where the $t$-th column corresponds to the activation set $\mathcal{R}(t)$, i.e.,  $R_{it}=1\iff i\in\mathcal R(t)$. Each entry of $\mat R$ is a realization of an independent Bernoulli random variable with probability equal to $p=0.4$.  We check if $\mat R\vect{1}\geq \vect{1}$, which guarantees Assumption \ref{a:activation} holds. Otherwise, we re-generate the matrix.

We start by analyzing the algorithm performance in terms of total control effort and action convergence time, comparing the constant control input proposed in Section~\ref{sec:constant}, and the feedback control policy proposed in Section~\ref{sec:feedback} on a real-world network with social contacts in a village in rural Malawi~\cite{malawi_net}, consisting of $n=84$ individuals and 346 weighted undirected edges. For feedback control, the simulations are carried out with Algorithm~\ref{alg:final} using revision step size $\Delta T=2$, so that $\mathcal T=\{\Delta T,2\Delta T,\dots\}$.  
As shown in Fig.~\ref{fig:Malawi01_comparison_c}, the feedback control policy has an initial phase in which the controlled nodes coincide with the constant strategy. Then, 
the number of controlled nodes decreases over time, until eventually no input is used,  reducing the total control effort. 
The red and green dots in Fig.~\ref{fig:Malawi01_comparison_c} indicate the time $T^*$ at which all nodes’ actions reaches the desired value $+1$ using the constant and feedback control policies, respectively. Interestingly, these values seem to coincide, implying that our feedback control policies allows to reduce the total control effort, without slowing down convergence. 

It is important to note that the control policy cannot be stopped at $T^*$, as the opinion equilibrium has not been reached (turning off the input may result in nodes switching their action back to $-1$). In Fig.~\ref{fig:Malawi01_comparison_y}, the convergence of the opinion vectors is reported for both constant and feedback control, evaluated as $\log (\|\vect{y}(t)-\vect{1}\|_\infty)$. While the constant control policy attains slightly faster convergence, the improvement is marginal and only when the opinions are already extremely close to the equilibrium $\vect{y^*}=\vect{1}$. 

These findings are  supported by  simulations  on larger synthetic networks with $n=200$ nodes; see Table \ref{tab:table_cost}. We compare the the action convergence time $T^*$ and the overall control cost $J$ up to time $\hat{T}_\varepsilon:=\min\{t:\|\vect{y}(t)-\vect{1}\|\leq 10^{-4}\}$ (note that $\hat{T}_\varepsilon$ could differ in constant and feedback control scenarios). Our feedback control policy reduces the cost by on average $5$\% on Erdős–Rényi (ER) graphs, without increasing the convergence time. Such reduction increases as the network becomes more structured: $26$\%  on Watts-Strogatz  (WS)  networks and up to $50\%$ on the real-world  network. The code is available at \href{github.com/robertaraineri/control-coev-ecc26}{https://github.com/robertaraineri/control-coev-cdc2026}.

\begin{figure}
 \centering
\subfloat[Number of controlled nodes]{
    \includegraphics[width=0.45\linewidth]{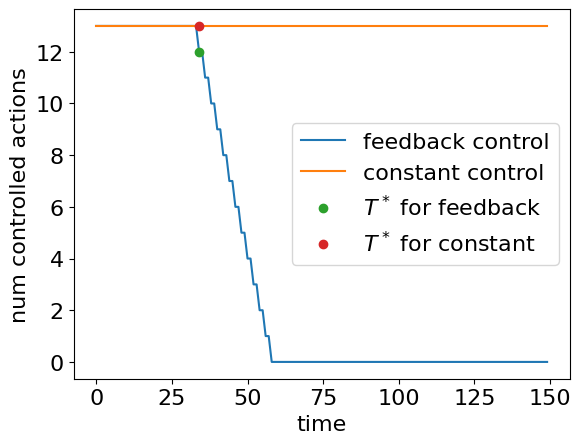}   \label{fig:Malawi01_comparison_c}}
\subfloat[Opinion convergence]{\includegraphics[width=0.48\linewidth]{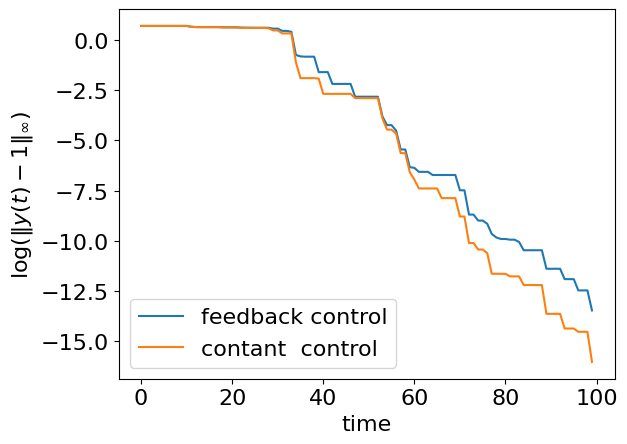}  \label{fig:Malawi01_comparison_y}}
   \caption{Feedback vs constant control policy: (a)  number of controlled and (b) distance from the desired opinion consensus; $(\lambda, \beta)=(0.4,0.6)$.}
    \label{fig:malawi_comparison}
\end{figure}

\begin{table}
    \centering
\begin{tabular}{>{\columncolor{gray!20}}c c c c c}
\hline
\rowcolor{gray!20}
& & \multicolumn{2}{c}{Feedback\,$\mid$\,\textit{Constant}} \\
\cline{3-4}
\rowcolor{gray!20}
Network & Parameters & Control Effort $J$ & $T^*$ \\
\hline
ER & $(\lambda, \beta)=(0.4,0.6)$ & $3314 \mid \it 3481$ & $28 \mid\it 28$ \\
ER & $(\lambda, \beta)=(0.3,0.7)$ & $3500 \mid\it 3640$ & $28\mid\it 28$\\
ER & $(\lambda, \beta)=(0.2,0.8)$ & $3400 \mid\it 3640$  & $28 \mid\it 28$\\
WS & $(\lambda, \beta)=(0.4,0.6)$ & $3718 \mid\it 3955$ & $75\mid\it 75$\\
WS & $(\lambda, \beta)=(0.3,0.7)$ & $4660 \mid\it 5236$ & $138\mid\it 138$\\
WS & $(\lambda, \beta)=(0.2,0.8)$ & $3456 \mid\it 4671$ & $100\mid\it 100$\\
\hline
\end{tabular}
    \caption{Comparison of performance in terms of the total control effort $J$ and the consensus time $T^*$. 
    Network parameters: ER connection probability $p=0.3$; WS average degree $k=40$ and rewiring probability  $p=0.4$. }
    \label{tab:table_cost}\vspace{-10pt}
\end{table}

\section{Conclusion}\label{sec:conclusion}

We 
focused on reaching a collective consensus by broadcasting to selected individuals, e.g. as done in realistic scenario via media sources or recommender agents. 
The main contribution of this work lies in extending the existing framework beyond the limitations of previous studies~\cite{raineri2025_tcns}, providing both a more general control-oriented formulation of the problem and a more effective strategy for intervention. 
These results open several  directions for future research. First, algorithms proposed for different domains suggest that sometimes it might be convenient to change add new target nodes or activate the input only at specific times~\cite{Zino2023fast}. Investigating these alternatives could further reduce the  control effort.
Second, the same framework could be adapted to pursue different objectives, e.g., preserving diversity.


\appendix



\begin{lemma}\label{lemma:y-monotonicity}
    Consider a controlled coevolutionary dynamics under Assumptions~\ref{a:activation}--\ref{a:constant}. For any $\tau > 0$, assume that $\vect{x}(t)\geq \vect{x}(t-1)$, $\forall\,t\leq \tau$. Then, $\vect{y}(t)\geq \vect{y}(t-1)$, $\forall\, t\leq \tau$. 
\end{lemma}
\begin{proof}
We prove the inequality component-wisely. For $i\in\mathcal C$, it holds $y_i(t)=(1-\lambda_i)+\lambda_ix_i(t-1)$.  Consequently, if $\vect x(t)$ is an increasing function, then the same holds for $\vect y(t)$. For $i\notin \mathcal C$, their dynamics coincides with the one from~\cite{raineri2024}. Thus, the monotonicity directly follows~\cite[Lemma 1]{raineri2024}.  
\end{proof}


\begin{lemma}\label{lemma:x-monotonicity}
    Consider a controlled coevolutionary dynamics under Assumptions~\ref{a:activation}--\ref{a:constant}. Then, $\vect{x}(t+1)\geq \vect{x}(t)$, $\forall\,t\geq 0$. 
\end{lemma}
\begin{proof}
We proceed by contradiction. Assume that $\vect{x}(t)$ is not monotonically non-decreasing. Then, there exists a finite time $\hat t := \inf \{t : \exists\,i \in \mathcal V \; \text{such that} \; x_i(t)<x_i(t-1) \}$. 
 By definition, at time $\hat t$ there is necessarily an $i\in\mathcal V$ for which $x_i(\hat t-1)=+1$ and $x_i(\hat t)=-1$, i.e., $\mathcal{A}(\hat t) \subset \mathcal{A}(\hat{t}-1)$. Proposition~\ref{prop:dynamics} indicates that $\delta_i(\hat t-1)<0$. Moreover, since $x_i(\hat t-1)=+1$, there must be a time  $\tilde t\in\{\hat t-1-T,\dots,\hat t-2\}$ in which $\delta_i(\tilde t)\geq 0$.  Hence, $\delta_i(\hat t-1)<\delta_i(\tilde t)$. Let us use the definition of $\delta_i(t)$ to rewrite the condition $\delta_i(\hat t-1)<\delta_i(\tilde t)$ and prove that this inequality is impossible. In fact, we get
\begin{equation}
\label{contrad_ineq}
\begin{array}{lc}
    & 2 \beta_i (1-\lambda_i)\left[(1-\gamma_i)\hspace{-.1cm}\sum\nolimits_{j \in \mathcal V} {w}_{ij} y_j(\tilde t)+\gamma_i\right] \hspace{-.05cm}+ \\
    &\hspace{-.05cm}(1- 
 \beta_i)\hspace{-.1cm}\sum\nolimits_{j \in \mathcal V}\! a_{ij} x_j(\tilde t )>\\
 &2\beta_i (1-\lambda_i)\left[ (1-\gamma_i)\sum_{j \in \mathcal V} {w}_{ij} y_j(\hat t -1)+\gamma_i \right]+\\
 & (1- 
 \beta_i) \sum_{j \in \mathcal V} a_{ij} x_j(\hat t -1).
 \end{array}
\end{equation}
Given that \eqref{contrad_ineq} involves only the values of $\vect{x}$ and $\vect{y}$ up to time $\hat t-1$ and $\hat t$ is the first time instant in which $\vect{x}(t)$ has decreased, then  for any $t\leq \hat t-1$ the action vector $\vect{x}(t)$ is monotonically non-decreasing. Hence by Lemma~\ref{lemma:y-monotonicity}, $\vect{y}(\tilde t)\leq \vect{y}(\hat t-1)$, and so $\sum_{j \in \mathcal V} {w}_{ij} y_j(\tilde t)\leq \sum_{j \in \mathcal V} {w}_{ij} y_j(\hat t -1)$. Thus, the first term on the LHS of \eqref{contrad_ineq} cannot exceed the corresponding RHS term. Hence, \eqref{contrad_ineq} may hold only if $\sum_{j \in \mathcal V} a_{ij} x_j(\tilde t) >\sum_{j \in \mathcal V} a_{ij} x_j(\hat t -1)$, which is impossible by the definition of $\hat t$. Hence, \eqref{contrad_ineq} cannot hold and so $\hat t$ cannot exist, yielding the contradiction and thus the claim.
\end{proof}

By combining Lemmas~\ref{lemma:y-monotonicity} and~\ref{lemma:x-monotonicity}, we conclude that $\vect{z}(t)=(\vect{x}(t),\vect y(t))$ is a monotonically non-decreasing function of $t$. Moreover, both  $\vect{x}(t)$ and $\vect y(t)$ are bounded by Proposition~\ref{prop:dynamics}. Hence, by the monotone convergence theorem~\cite{bartle1976}, both sequences $\vect y(t)$ and $\vect x(t)$ admit a limit. Being $\vect x(t)$ a discrete (and finite) sequence, its convergence should necessary occur in finite time, while convergence of $\vect y(t)$ is asymptotic. \qed



\begin{thebibliography}{21}
\providecommand{\url}[1]{#1}
\csname url@samestyle\endcsname
\providecommand{\newblock}{\relax}
\providecommand{\bibinfo}[2]{#2}
\providecommand{\BIBentrySTDinterwordspacing}{\spaceskip=0pt\relax}
\providecommand{\BIBentryALTinterwordstretchfactor}{4}
\providecommand{\BIBentryALTinterwordspacing}{\spaceskip=\fontdimen2\font plus
\BIBentryALTinterwordstretchfactor\fontdimen3\font minus \fontdimen4\font\relax}
\providecommand{\BIBforeignlanguage}[2]{{%
\expandafter\ifx\csname l@#1\endcsname\relax
\typeout{** WARNING: IEEEtran.bst: No hyphenation pattern has been}%
\typeout{** loaded for the language `#1'. Using the pattern for}%
\typeout{** the default language instead.}%
\else
\language=\csname l@#1\endcsname
\fi
#2}}
\providecommand{\BIBdecl}{\relax}
\BIBdecl

\bibitem{centola2005emperor}
D.~Centola, R.~Willer, and M.~Macy, ``The emperor’s dilemma: A computational model of self-enforcing norms,'' \emph{Am. J. Sociol.}, vol. 110, no.~4, pp. 1009--1040, 2005.

\bibitem{friedkin2015_socialsurvey}
N.~E. Friedkin, ``{The Problem of Social Control and Coordination of Complex Systems in Sociology: A Look at the Community Cleavage Problem},'' \emph{IEEE Control Syst. Mag.}, vol.~35, no.~3, pp. 40--51, 2015.

\bibitem{proskurnikov2017tutorial}
A.~V. Proskurnikov and R.~Tempo, ``{A tutorial on modeling and analysis of dynamic social networks. Part I},'' \emph{Annu. Rev. Control}, vol.~43, pp. 65--79, 2017.

\bibitem{montanari2010spread_innovation}
A.~Montanari and A.~Saberi, ``The spread of innovations in social networks,'' \emph{Proc. Natl. Acad. Sci. USA}, vol. 107, no.~47, pp. 20\,196--20\,201, 2010.

\bibitem{gavrilets2017collective}
S.~Gavrilets and P.~J. Richerson, ``Collective action and the evolution of social norm internalization,'' \emph{Proc. Natl. Acad. Sci. USA}, vol. 114, no.~23, pp. 6068--6073, 2017.

\bibitem{Martins2008coda}
A.~C.~R. Martins, ``Continuous opinions and discrete actions in opinion dynamics problems,'' \emph{Int. J. Mod. Phys. C}, vol.~19, no.~04, pp. 617--624, 2008.

\bibitem{Ceragioli2018quantized}
F.~Ceragioli and P.~Frasca, ``Consensus and disagreement: The role of quantized behaviors in opinion dynamics,'' \emph{SIAM J. Control Optim.}, vol.~56, no.~2, pp. 1058--1080, 2018.

\bibitem{zino2020chaos}
L.~Zino, M.~Ye, and M.~Cao, ``A two-layer model for coevolving opinion dynamics and collective decision-making in complex social systems,'' \emph{Chaos}, vol.~30, no.~8, p. 083107, 2020.

\bibitem{Zino2020cdc}
------, ``A coevolutionary model for actions and opinions in social networks,'' in \emph{59th IEEE Conf. Decis. Control}, 2020, pp. 1110--1115.

\bibitem{Hassan2023tac}
H.~D. Aghbolagh, M.~Ye, L.~Zino, Z.~Chen, and M.~Cao, ``Coevolutionary dynamics of actions and opinions in social networks,'' \emph{IEEE Trans. Autom. Control.}, vol.~68, no.~12, pp. 7708--7723, 2023.

\bibitem{raineri2025_tcns}
R.~Raineri, M.~Ye, and L.~Zino, ``Controlling a social network of individuals with coevolving actions and opinions,'' \emph{IEEE Trans. Control Netw. Syst.}, 2026.

\bibitem{Li2020}
T.~Li and H.~Zhu, ``Effect of the media on the opinion dynamics in online social networks,'' \emph{Physica A}, vol. 551, p. 124117, Aug. 2020.

\bibitem{Rossi2022}
W.~S. Rossi, J.~W. Polderman, and P.~Frasca, ``The closed loop between opinion formation and personalized recommendations,'' \emph{IEEE Trans. Control Netw. Syst.}, vol.~9, no.~3, pp. 1092--1103, 2022.

\bibitem{Sprenger2024}
B.~Sprenger, G.~De~Pasquale, R.~Soloperto, J.~Lygeros, and F.~Dörfler, ``Control strategies for recommendation systems in social networks,'' \emph{IEEE Control Syst. Lett.}, vol.~8, pp. 634--639, 2024.

\bibitem{ye2021nat}
M.~Ye \emph{et~al.}, ``Collective patterns of social diffusion are shaped by individual inertia and trend-seeking,'' \emph{Nat. Comm.}, vol.~12, p. 5698, 2021.

\bibitem{friedkin1990_FJsocialmodel}
N.~E. Friedkin and E.~C. Johnsen, ``{Social Influence and Opinions},'' \emph{J. Math. Sociol.}, vol.~15, no. 3-4, pp. 193--206, 1990.

\bibitem{raineri2024}
R.~Raineri, G.~Como, F.~Fagnani, M.~Ye, and L.~Zino, ``On controlling a coevolutionary model of actions and opinions,'' \emph{63rd IEEE Conf. Decis. Control}, pp. 4550--4555, 2024.

\bibitem{Como2022supermodular}
G.~Como, S.~Durand, and F.~Fagnani, ``Optimal targeting in super-modular games,'' \emph{IEEE Trans. Autom. Control.}, vol.~67, no.~12, pp. 6366--6380, 2022.

\bibitem{malawi_net}
L.~Ozella \emph{et~al.}, ``Using wearable proximity sensors to characterize social contact patterns in a village of rural {Malawi},'' \emph{EPJ Data Sci.}, vol.~10, no.~1, 2021.

\bibitem{Zino2023fast}
L.~Zino, G.~Como, and F.~Fagnani, ``Fast spread in controlled evolutionary dynamics,'' \emph{IEEE Trans. Control Netw. Syst.}, vol.~10, no.~3, pp. 1555--1567, 2023.

\bibitem{bartle1976}
R.~G. Bartle, \emph{{The Elements of Real Analysis}}.\hskip 1em plus 0.5em minus 0.4em\relax Wiley, 1976.

\end{thebibliography}
\end{document}